\documentclass[letterpaper, 10 pt, conference]{ieeeconf}  % Comment this line out if you need a4paper

\IEEEoverridecommandlockouts                              % This command is only needed if 
                                                                                       
                                               \usepackage{amsmath,graphicx,multicol}
\usepackage{amsfonts}
\usepackage{color}
\usepackage{bbm}
\usepackage{cite}
\usepackage{amssymb}
\usepackage{algorithm}
\usepackage{algorithmic}
\usepackage{tabularx}
\usepackage[utf8]{inputenc}
\usepackage[english]{babel}

\usepackage{amsthm}
\usepackage{subcaption}
\usepackage[font={footnotesize}]{caption}
\usepackage{url}
\usepackage{mathtools}
\usepackage{comment}
\usepackage[keeplastbox]{flushend} %To make it balance
\newcommand{\norm}[1]{\left|\left|#1\right|\right|}
\newcommand{\Real}{{\mathbbm{R}}}

\newcommand{\sI}{{\cal{L}}}
\newcommand{\sO}{{\cal{O}}}
\newcommand{\rN}{\mathcal{N}}

\newcommand{\vx}{\mathbf{x}}
\newcommand{\vy}{\mathbf{y}}
\newcommand{\vn}{\mathbf{n}}
\newcommand{\vw}{\mathbf{w}}
\newcommand{\vh}{\mathbf{h}}
\newcommand{\vP}{\mathbf{P}}
\newcommand{\vH}{\mathbf{H}}
\newcommand{\vA}{\mathbf{A}}
\newcommand{\vQ}{\mathbf{Q}}
\newcommand{\vR}{\mathbf{R}}
\newcommand{\vI}{\mathbf{I}}

\DeclareMathOperator{\E}{\mathbb{E}}
\newcommand{\argmax}{\arg\!\max}

\newtheorem{definition}{Definition}[]
\newtheorem{proposition}{Proposition}[]
\newtheorem{theorem}{Theorem}[]

\def\MSE{{\mathrm{MSE}}}

\def\P{{\mathbf P}}

\def\I{{\mathbf I}}

\def\F{{\mathbf F}}

\def\x{{\mathbf x}}

\def\A{{\mathbf A}}

\def\P{{\mathbf P}}

\def\R{{\mathbbm{R}}}
\def\I{{\mathbf I}}
\def\a{{\mathbf a}}

\def\L{{\cal L}}
\def\S{{\cal S}}
\def\T{{\cal T}}

\def\X{{\cal X}}
\newcommand{\ts}{\textsuperscript}

\let\oldref\ref
\renewcommand{\ref}[1]{(\oldref{#1})}
\newcommand{\RNum}[1]{\uppercase\expandafter{\romannumeral #1\relax}}
\makeatletter
\renewcommand{\fnum@figure}{Fig.~\thefigure}
\makeatother

\title{\LARGE \bf
Near-Optimal Distributed Estimation for a Network of Sensing Units Operating Under Communication Constraints
}
\author{Abolfazl Hashemi$^{\ast}$, Osman Fatih Kilic$^{\ast}$, and Haris Vikalo$^{}$ %
\thanks{Authors are with the Department of Electrical and Computer Engineering, University of Texas at Austin, Austin, TX 78712 USA.}%
\thanks{$\ast$These authors contributed equally to the manuscript.}
}
\begin{document}
%\ninept
%
\maketitle
\begin{abstract}
We study the problem of distributed state estimation in a network of sensing units that can exchange their measurements but the rate of communication between the units is constrained. The units collect noisy, possibly only partial observations of the unknown state; they are assisted by a relay center which can communicate at a higher rate and schedules the exchange of measurements between the units. We consider the task of minimizing the total mean-square estimation error of the network while promoting balance between the individual units' performances. This problem is formulated as the maximization of a monotone objective function subject to a cardinality constraint. By leveraging the notion of weak submodularity, we develop an efficient greedy algorithm for the proposed formulation and show that the greedy algorithm achieves a constant factor approximation of the optimal objective.
Our extensive simulation studies illustrate the efficacy of the proposed formulation and the greedy algorithm.
\end{abstract}
%
%%%%%%%%%%%%%%%%%%%%%%%%%%%%%%%%%%%%%%%%%%%%%%%%%%%%%%%%
%%%%%%%%%%%%%%%%%%%%%%%Introduction%%%%%%%%%%%%%%%%%%%%%
%%%%%%%%%%%%%%%%%%%%%%%%%%%%%%%%%%%%%%%%%%%%%%%%%%%%%%%%
\section{Introduction}\label{sec:intro}
The problem of distributed estimation in a network of sensing units that are capable of exchanging information arises in a variety of settings. An example is the network of autonomous vehicles equipped with a number of sensors that enable tasks such as identification of the navigation paths and estimation of the position, velocity, and trajectory of nearby objects. However, state-of-the-art sensing technologies including radar, cameras, and LIDAR have limited sensing range and typically only provide information about objects in the line-of-sight. To overcome the lack of adequate information in cluttered and partially observed environments, autonomous vehicles communicate with centralized relay centers and other vehicles using vehicle-to-infrastructure and vehicle-to-vehicle communication protocols. An autonomous vehicle can generate up to one TB of data in a single trip \cite{lu2014connected}. Handling enormous amounts of sensing data in a network of autonomous vehicles presents a major challenge for the current communication technologies such as the dedicated short-range communication (DSRC) schemes \cite{kenney2011dedicated} that are characterized by limited range and transmission rates.

Given a network of units, it is generally of interest to design an inference scheme that minimizes the overall estimation error; however, in many applications it is of critical importance that each unit generates a reliable estimate so as not to adversely affect decision making of other units in the network (e.g., in the context of autonomous vehicles, a unit with high estimation error may need to slow down and force other units to do the same). Therefore, we are interested in minimizing the total mean-square estimation error for the entire network while promoting balanced performance of the individual units.

%%%%%%%%%%%%%%%%%%%%%%%%%%%%%%%%%%%%%%%%%%%%%%%%%%%%%%%%%%%%
%%%%%%%%%%%%%%%%%%%%%%%%%%%%%%%%%%%%%%%%%%%%%%%%%%%%%%%%%%%%
%\subsection{Sensor selection}
The well-known sensor selection problem \cite{krause2008near,joshi2009sensor,shamaiah2010greedy,shamaiah2012greedy,tzoumas2016sensor,ma} can be thought of as a special instance of the described task. There, one is interested in the design of an optimal estimation scheme under communication constraints for a single unit (i.e., the fusion center) which collects sensor data. More specifically, due to various practical considerations and limitations on resources including computational and communication constraints, the fusion center typically aggregates information by querying only a small subset of the available sensors. Since finding an optimal solution to the sensor selection problem is NP-hard, state-of-the-art sensor selection algorithms attempt to find an approximate solution in an iterative fashion by leveraging a greedy heuristic. For instance, \cite{shamaiah2010greedy,shamaiah2012greedy,tzoumas2016sensor} consider a greedy algorithm for the $\log \det$ maximization formulation of the sensor selection problem. Since the $\log \det$ of the Fisher information matrix is a monotone submodular function \cite{nemhauser1978analysis}, the greedy scheme developed in \cite{shamaiah2010greedy,shamaiah2012greedy,tzoumas2016sensor} is a $(1-1\slash e)$-approximation algorithm. A randomized greedy approach proposed in \cite{ma} leverages the weak submodularity of the mean-square error (MSE) objective in the sensor selection problem. Greedy sensor selection solvers are employed in various related problems in control systems, signal processing, and machine learning. Examples include sensor selection for Kalman filtering \cite{nordio2015sensor,shamaiah2010greedy,tzoumas2016sensor}, batch state estimation and stochastic process estimation \cite{tzoumas2016scheduling,tzoumas2016near}, minimal actuator placement \cite{tzoumas2015minimal}, subset selection in machine learning \cite{mirzasoleiman2014lazier}, voltage control and meter placement in power networks \cite{gensollen2016submodular,liu2016towards}, and sensor scheduling in wireless sensor networks \cite{shamaiah2012greedy,nordio2015sensor}. None of these methods, nor the related distributed and consensus-based schemes in \cite{mirzasoleiman2016distributed,das2015distributed,battistelli2015consensus}, are directly applicable for the setting that we consider: estimation in resource-constrained networks of sensing units with the goal of simultaneously minimizing the total MSE of the network while promoting balanced performance of the individual units.
%%%%%%%%%%%%%%%%%%%%%%%%%%%%%%%%%%%%%%%%%%%%%%%%%%%%%%%%%%%%
%%%%%%%%%%%%%%%%%%%%%%%%%%%%%%%%%%%%%%%%%%%%%%%%%%%%%%%%%%%%
%\subsection{Submodular maximization}
%%%%%%%%%%%%%%%%%%%%%%%%%%%%%%%%%%%%%%%%%%%%%%%%%%%%%%%%%%%%
%%%%%%%%%%%%%%%%%%%%%%%%%%%%%%%%%%%%%%%%%%%%%%%%%%%%%%%%%%%%
%\subsection{Contribution}

In this paper, we address the above challenges by making the following key contributions:
\begin{itemize}
\item  We formulate the task of state estimation in a network of sensing units under a constraint on communication resources and a demand for balanced performance of the individual units as the problem of maximizing a monotone objective function subject to cardinality constraint. The cardinality constraint naturally captures the aforementioned communication constraint. The proposed objective function consists of two parts: the total MSE of the network and a regularizing term that promotes balanced performance of individual units.
 \item We develop an efficient greedy algorithm for the proposed NP-hard formulation. By leveraging the notion of weak submodularity, we show that the greedy algorithm achieves a constant factor approximation of the optimal schedule.
\item In simulation studies, we illustrate that our proposed formulation promotes balanced performance of the individual units while minimizing the total MSE of the network.
\end{itemize}
%The rest of the paper is organized as follows. Section \oldref{sec:sys} explains the system model and reviews some related concepts. In Section \oldref{sec:model}, we introduce the proposed optimization framework. The proposed greedy algorithm and its performance analysis are discussed in Section \oldref{sec:alg}. Section \oldref{sec:sim} presents the simulation results while the concluding remarks are stated in Section \oldref{sec:concl}.
%%%%%%%%%%%%%%%%%%%%%%%%%%%%%%%%%%%%%%%%%%%%%%%%%%%%%%%%
%%%%%%%%%%%%%%%%%%%%%%%Notations&stuff%%%%%%%%%%%%%%%%%%
%%%%%%%%%%%%%%%%%%%%%%%%%%%%%%%%%%%%%%%%%%%%%%%%%%%%%%%%
\section{Problem Description}\label{sec:sys}
\subsection{Notation and preliminaries}
First, we briefly summarize the notation used in the paper. Sets are denoted by calligraphic letters,  $[n] := \{1,2,\dots,n\}$, and $|\S|$ denotes the cardinality of set $\S$. Bold capital letters are used to denote matrices while bold lowercase letters represent column vectors. $\A_{ij}$ denotes the $(i,j)$ entry of $\A$, $\a_j$ is the $j\ts{th}$ row of $\A$, $\A_{\S}$ is a submatrix of $\A$ that contains rows indexed by the set $\S$, and $\lambda_{max}(\A)$
and $\lambda_{min}(\A)$ are the largest and the smallest eigenvalues of $\A$, respectively. Finally, $\I_n \in \R^{n\times n}$ is the identity matrix.

%%%%%%%%%%%%%%%%%%%%%%%%%%%%%%%%%%%%%%%%%%%%%%%%%%%%%%%%%%%%
%%%%%%%%%%%%%%%%%%%%%%%%%%%%%%%%%%%%%%%%%%%%%%%%%%%%%%%%%%%%
Next, we overview some definitions that are essential in the development and analysis of the proposed framework.
\begin{definition}
\label{def:submod}
Set function $f:2^\X\rightarrow \mathbb{R}$ is submodular if 
\begin{equation*}
f(\S\cup \{j\})-f(\S) \geq f(\T\cup \{j\})-f(\T)
\end{equation*}
for all subsets $\S\subseteq \T\subset \X$ and $j\in \X\backslash \T$. The term $f_j(\S)=f(\S\cup \{j\})-f(\S)$ is the marginal value of adding element $j$ to set $\S$. Furthermore, $f$ is monotone if $f(\S)\leq f(\T)$ for all $\S\subseteq \T\subseteq \X$.
\end{definition}

\begin{definition}
The maximum element-wise curvature of a monotone non-decreasing function $f$ is defined as
 \begin{equation*}
 {\cal C}_{f}=\max_{l\in [n-1]}{\max_{(\S,\T,i)\in \mathcal{\X}_l}{f_i(\T)\slash f_i(\S)}},
 \end{equation*}
where $\mathcal{\X}_l = \{(\S,\T,i)|\S \subset \T \subset \X, i\in \X \backslash \T, |\T\backslash \S|=l,|\X|=n\}$. 
%Furthermore, the maximum element-wise curvature is given by ${\cal C}_{\max}=\max_{l=1}^{n-1}{{\cal C}_l}$.
\end{definition}
The maximum element-wise curvature is a closely related concept to submodularility and essentially quantifies  how close the set function is to being submodular. It is worth noting a set function $f(\S)$ is submodular if and only if its maximum element-wise curvature satisfies 
${\cal C}_{f} \le 1$.

% Next, we state the problem of information distribution in a network regarding MSE performance of the agents and the network. Then, we propose a greedy method to schedule the information distribution that results in a near-optimal performance.
%%%%%%%%%%%%%%%%%%%%%%%%%%%%%%%%%%%%%%%%%%%%%%%%%%%%%%%%%%%%
%%%%%%%%%%%%%%%%%%%%%%%%%%%%%%%%%%%%%%%%%%%%%%%%%%%%%%%%%%%%
\subsection{System model}
\begin{figure}[t]
\includegraphics[width=3.5in]{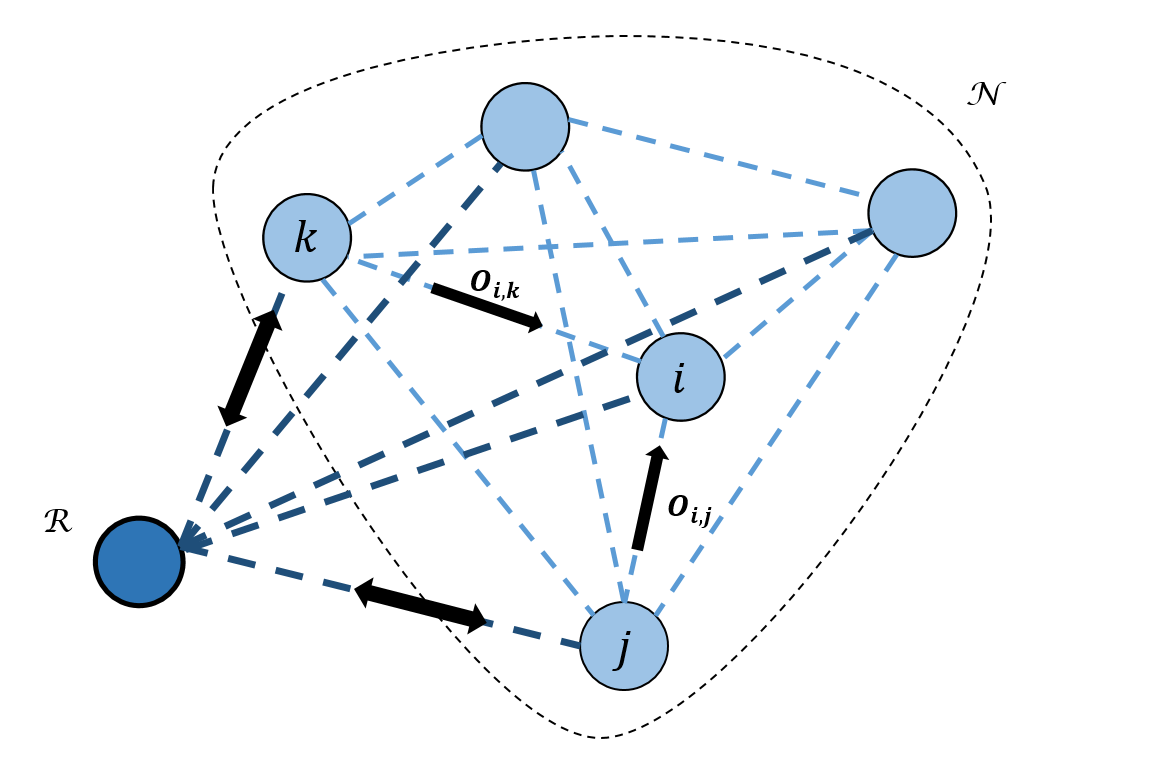}
\centering
\caption{A fully connected network of units with sensing, communication, and processing capabilities; the communication between the units is constrained. Also shown is a relay node $R$ that schedules exchange of observations $\sO_{i,j}$ and  $\sO_{i,k}$ to node $i$ from nodes $j$ and $k$, respectively.}
\label{pd:network}
\vspace{-0.5cm}
\end{figure}
%%%%%%%%%%%%%%%%%%%%%%%%%%%%%%
We consider a fully connected distributed network of $m$ nodes with sensing, communication, and processing capabilities.\footnote{Throughout the paper, the words unit and node are used interchangeably.}
The network also includes a relay node that schedules the exchange of information among the units, i.e., the relay node decides which information should be communicated from one unit to another. We assume that the communication among units is constrained but the channels between units and the relay node have large bandwidth (e.g., the vehicle-to-vehicle and vehicle-to-infrastructure communication settings, respectively \cite{miller2008vehicle}).

One can think of the described network as having an undirected graph structure, where edges and vertices represent the nodes and the connections among them, respectively. An example of such a network is illustrated in Fig. \oldref{pd:network}. There, the relay node $R$ schedules exchange of observations $\sO_{i,j}$ and  $\sO_{i,k}$ to node $i$ from nodes $j$ and $k$, respectively. 

To model the dynamics of the underlying hidden state $\vx \in \R^n$, we assume a state-space model
\begin{align*}
\vx_{t+1} = \vA_{t} \vx_t + \vw_t,
\end{align*}
where $\vA_t \in \R^{n\times n}$  is the state-transition matrix and $\vw_t \in \R^n$ is the zero-mean Gaussian state noise with covariance $\vQ \in \Real^{n \times n}$. We further assume that the state $\x_t$ is uncorrelated with $\vw_t$ and the initial state $\vx_0$ is sampled from a Gaussian distribution, i.e., $\x_0 \sim \rN(0,\mathbf{\Sigma}_x)$. 

The $i\ts{th}$ node in the network acquires partial noisy linear observations of the underlying state according to
\begin{align*}
\vy_{i,t} = \vH_{i,t}\vx_t + \vn_{i,t},
 \end{align*}
where $i\in [m]$ and $\vH_{i,t}$ denotes the matrix that selects observable components of the underlying state. Let $\sI_{i,t}$ denote the set of noisy observations of the components of $\vx_t$ available to the $i^{th}$ node (i.e., the noisy observations collected by the vector $\vy_{i,t}$). Here, we do not make any assumptions on the structure of 
$\vH_{i,t}$.\footnote{As we proceed, for the simplicity of notation we may omit the time index.} We assume that the observation noise $\vn_{i,t}\in\Real^{|\sI_{i,t}|}$ is spatially and temporally independent zero-mean Gaussian noise with covariance $\vR_{i,t} = \sigma_i^2\vI_{|\sI_{i,t}|}$.

Without communication, each node only uses its acquired local measurements and performs Kalman filtering to estimate the underlying state $\vx_t$ by minimizing the mean-squared error of the linear least-mean square error (LLMSE) estimator. However, cooperation can greatly enhance the learning capabilities of the individual units as well as the entire network.

Let $\vP_{\sI_i,t}$ be the filtered error covariance matrix of the $i\ts{th}$ agent at time $t$ obtained by using only the local measurements $\L_{i,t}$.
%, and $\vP_{i,t}$ be the prediction error, 
Then,
\begin{align}
%\begin{split}
\P_{\sI_i,t} = \left(\vP_{i,t-1}^{-1}+\frac{1}{\sigma_i^2} \sum_{i_j\in\sI_{i,t}} \vh_{i_j} \vh_{i_j}^\top\right)^{-1} %\\
%\vP_{i,t} &= \vA_t P_{\sI_i,t} \vA_t^\top + \vQ.
%\end{split}
\end{align}
where 
\begin{align}
\vP_{i,t} = \vA_t \P_{\sI_i,t} \vA_t^\top + \vQ
\end{align}
is the prediction error covariance matrix. If the relay node $R$ at time $t$ allocates observations to node $i$, agent $i$ receives the subset $\sO_i$ of the measurements from the agents selected by $R$. Let
\begin{equation}
\F_{i,t} = \vP_{i,t-1}^{-1}+\frac{1}{\sigma_i^2} \sum_{i_j\in\sI_{i,t}} \vh_{i_j} \vh_{i_j}^\top
\end{equation}
be the Fisher information matrix associated with the $i^{th}$ node that determines the prior information and confidence of the node $i$ before receiving the partial observation set $\sO_i$. The filtered error covariance matrix of the $i^{th}$ node will then be updated according to
\begin{align}
 \vP_{\sI_i \cup \sO_i} = \left(\F_{i,t} + \sum_{j_k\in \sO_i} \frac{1}{\sigma_j^2}\vh_{j_k}\vh_{j_k}^\top\right)^{-1}.
\end{align}
The global MSE of the network at time $t$ is defined as the sum of the MSEs of the individual nodes. In particular, 
\begin{align}
\MSE_t = \sum_{i=1}^m \E{\norm{\x_t-\hat{\x}_{i,t}}^2},
\end{align}
where $\hat{\x}_{i,t}$ is the linear estimate of $\x_t$ computed by the $i^{th}$ unit at time $t$. Since the MSE is equivalent to the trace of the filtered error covariance matrix,
\begin{align}
\MSE_t = \sum_{i=1}^m \mathrm{Tr}\left(\left(\F_{i,t} + \sum_{j_k\in \sO_i} \frac{1}{\sigma_j^2}\vh_{j_k}\vh_{j_k}^\top\right)^{-1}\right).
\end{align}

As stated in Section \oldref{sec:intro}, communication constraints limit the amount of information that can be exchanged among the nodes of the network at any given time step. More specifically, we assume that the subsets of partial observations $\{\sO_i\}_{i=1}^m$ scheduled to be communicated to each agent should satisfy $\sum_{i=1}^m |\sO_i| \leq K$,
% \begin{align*}
% \sum_{i=1}^m |\sO_i| \leq K,
% \end{align*}
where $K$ denotes the total number of observations that are allowed to be exchanged among the nodes of the network.
The relay node decides how to allocate measurements to individual nodes by solving the optimization problem
\begin{equation}\label{eq:sensor}
\begin{aligned}
& \underset{\sO_i}{\text{min}}
\quad \sum_{i=1}^m \mathrm{Tr}\left(\left(\F_{i,t} + \sum_{j_k\in \sO_i} \frac{1}{\sigma_k^2}\vh_{j_k}\vh_{j_k}^\top\right)^{-1}\right) \\
& \text{s.t.}\hspace{0.5cm}  \sO_i \subset \cup_{i=1}^m\L_i\text{,} \quad \forall i\in [m] \\
& \quad\hspace{0.5cm}  \sum_{i=1}^m |\sO_i| \leq K.
\end{aligned}
\end{equation}
A comparison of \ref{eq:sensor} to a (simpler) sensor selection problem reveals that finding the optimal solution to \ref{eq:sensor} is generally NP-hard. In addition to being computationally challenging, optimization \ref{eq:sensor} does not necessary lead to a solution that would promote balanced MSE performance of the individual units; this point is illustrated by the simulation results in Section \oldref{sec:sim}. To this end, we next add a regularization term to the objective function so as to promote balanced performance while still finding a near-optimal solution to the MSE estimation problem for the entire network.
%%%%%%%%%%%%%%%%%%%%%%%%%%%%%%%%%%%%%%%%%%%%%%%%%%%%%%%%
%%%%%%%%%%%%%%%%%%%%%%%%%%%Algorithm%%%%%%%%%%%%%%%%%%%%
%%%%%%%%%%%%%%%%%%%%%%%%%%%%%%%%%%%%%%%%%%%%%%%%%%%%%%%%
\section{Promoting balanced performance of the  units}\label{sec:model} 
Let $\S \subseteq \X$ where $\X = [m]\times [m] \times [\max_{i\in [m]}|\L_i|]$ is a ground set for the set function
\begin{multline}\label{eq:f}
f(\S) = \sum_{i=1}^m \mathrm{Tr}\left(\F_{i,t}^{-1}\right) \\ \hspace{2cm }- \mathrm{Tr}\left(\left(\F_{i,t} + \sum_{(i,j,k)\in \S} \frac{1}{\sigma_{j}^2}\vh_{j_k}\vh_{j_k}^\top\right)^{-1}\right).
\end{multline}
The triplet $(i,j,k)$ denotes that the $k\ts{th}$ measurement of node $j$ is communicated to node $i$. The function $f(\S)$ is inversely related to the total MSE of the network. To arrive at a measurement exchange scheme that promotes balanced performance across the network units, we propose the optimization problem
\begin{equation}\label{eq:balance}
\begin{aligned}
& \underset{\S}{\text{max}}
\quad f(\S)+ \gamma g(\S) \\
& \text{s.t.}\hspace{0.5cm}  \sO_i \subset \cup_{i=1}^m\L_i\text{,} \quad \forall i\in [m] \\
& \quad\hspace{0.5cm}  \sum_{(i,-,-)\in \S} |\sO_i| \leq K\\
& \quad\hspace{0.5cm} \S \subseteq [m]\times [m] \times [\max_{i\in [m]}|\L_i|],
\end{aligned}
\end{equation}
where
\begin{equation}\label{eq:g}
g(S) = \sum_{(i,-,-) \in \S} \log\left(1+\frac{|\sO_i|}{|\L_i|}\right)
\end{equation}
is a regularization function and $\gamma\geq 0$ denotes the regularization parameter that determines the significance of balancing with respect to the goal of minimizing the total MSE of the entire network. On one hand, when $\gamma = 0$ the relay node $R$ attempts to find a schedule that results in the lowest total MSE  while disregarding potential imbalance in performance of the individual units. On the other hand, when $\gamma$ is relatively large, the exchange of information determined by $R$ is such that the differences between the MSEs of individual sensing nodes in the network become as small as possible. Notation $(i,-,-) \in \S$ in \ref{eq:balance} and \ref{eq:g} implies that it does not matter for $g(\S)$ which measurements are communicated to the $i^{th}$ node; instead, it is the number of communicated measurements that is used to promote balanced performance.

Note that the proposed formulation \ref{eq:balance} is an NP-hard combinatorial optimization problem, as it generalizes \ref{eq:sensor}. However, as we show next, the proposed objective function $u(\S) = f(\S)+ \gamma g(\S)$ is monotone weak submodular, i.e., under some mild conditions it is characterized with a bounded maximum element-wise curvature. Hence, one can find an approximate solution to \ref{eq:balance} using a greedy algorithm, as we state in the next section.

We proceed by providing two propositions to characterize the combinatorial properties of $f(\S)$ and $g(\S)$. For simplicity of the stated results, we assume that $\vR_{i,t} = \sigma^2\vI_{|\sI_{i,t}|}$, i.e., use the same measurement noise statistics for all sensing nodes of the network (a generalization is straightforward).
%%%%%%%%%%%%%%%%%%%%%%%%%%%%%%%%%%%%%%%%%%%%%%
%%%%%%%%%%%%%%%%%%%%%%%%%%%%%%%%%%%%%%%%%%%%%%
\begin{proposition}
Define $\lambda_M = \max_{i\in [m]}\lambda_{max}(\F_{i,t})$, $\lambda_m =  \min_{i\in [m]}\lambda_{min}(\F_{i,t})$, and $\vH_t = [\vH_{1,t}^\top,\dots,\vH_{1,m}^\top]^\top$. Let $\mathcal{C}_f$ be the maximum element-wise curvature of $f(\S)$. If
\begin{equation}\label{eq:cond1}
\frac{1}{\sigma^2}\lambda_{max}(\vH_t^\top\vH_t) \leq \lambda_M,
\end{equation}
then it holds that 
\begin{equation}
\mathcal{C}_f \leq \left(2\frac{\lambda_M }{\lambda_m }\right)^3.
\end{equation}
\end{proposition}
\begin{proof}
Note that $f(\S)$ is the sum of the additive inverse of the MSE of the sensing nodes that receive partial observations. Let $i$ be one such node. Theorem 1 in \cite{ma} states that if
\begin{equation}\label{eq:cond2}
\frac{1}{\sigma^2}\lambda_{max}(\vH_{i,t}^\top\vH_{i,t}) \leq \lambda_{max}(\F_{i,t}),
\end{equation}
then the maximum element-wise curvature of the additive inverse of the MSE of node $i$, $\mathcal{C}_i$, satisfies 
\begin{equation}
\begin{aligned}
\mathcal{C}_i &\leq \left(2\frac{\lambda_{max}(\F_{i,t})}{\lambda_{min}(\F_{i,t}) }\right)^3 \frac{1+\sigma^2\lambda_{min}(\F_{i,t})}{1+2\sigma^2\lambda_{max}(\F_{i,t})} \\
&\leq \left(2\frac{\lambda_{max}(\F_{i,t})}{\lambda_{min}(\F_{i,t}) }\right)^3.
\end{aligned}
\end{equation}
It is straightforward to see that the condition stated in \ref{eq:cond1} implies  \ref{eq:cond2} and we have $\max_{i\in [m]} \mathcal{C}_i\leq (2\frac{\lambda_M }{\lambda_m })^3 $. Hence, definition of $\lambda_M$, $\lambda_m$, and $f(\S)$ yields $\mathcal{C}_f \leq \max_{i\in [m]} \mathcal{C}_i$ which in turn completes the proof.
\end{proof}
%%%%%%%%%%%%%%%%%%%%%%%%%%%%%%%%%%%%%%%%%%%%%%
%%%%%%%%%%%%%%%%%%%%%%%%%%%%%%%%%%%%%%%%%%%%%%
\begin{proposition}
The set function $g(\S)$ is a monotone submodular function.
\end{proposition}
\begin{proof}
In order to prove the results, we first find the marginal gain $g_{(i,-,-)}(\S)$ that in the following argument is denoted by $g_{i}(\S)$ (with a slight abuse of notations for the sake of readability). By the definition of $g(\S)$ and the marginal gain,
\begin{equation}
\begin{aligned}
g_{i}(\S) &= \log\left(1+\frac{|\sO_i|+1}{|\L_i|}\right) - \log\left(1+\frac{|\sO_i|}{|\L_i|}\right)\\
& = \log\left(1+\frac{1}{|\sO_i|+|\L_i|}\right).
\label{eq:g_marginal}
\end{aligned}
\end{equation}
Since $\log(.)$ is a monotonically increasing function, $|\sO_i|\geq 0$, $|\L_i|> 0$, $g_{i}(\S)>0$ and hence $g(\S)$ is monotone. We now prove the second part of the statement, i.e., submodularity of $g(\S)$. Specifically, we should prove that the marginal gain of adding $(i,j,k)$ to $\S$ is greater than adding it to a larger set $\S\cup\{(i',j',k')\}$ where $(i,j,k) \neq (i',j',k')$. Two cases might happen. First, assume that $i \neq i'$. Then, 
\begin{equation}\label{eq:mg1}
g_{i}(\S) = g_{i}(\S\cup\{(i',j',k')\}) = \log\left(1+\frac{1}{|\sO_i|+|\L_i|}\right).
\end{equation}
Now assume $i = i'$. Then,
\begin{equation}\label{eq:mg2}
g_{i}(\S\cup\{(i',j',k')\}) = \log\left(1+\frac{1}{|\sO_i|+|\L_i|+1}\right) < g_{i}(\S).
\end{equation}
Combining \ref{eq:mg1} and \ref{eq:mg2} we conclude $g_{i}(\S) \geq g_{i}(\S\cup\{(i',j',k')\})$ which in turn implies submodularity.
\end{proof}
%%%%%%%%%%%%%%%%%%%%%%%%%%%%%%%%%%%%%%%%%%%%%%
%%%%%%%%%%%%%%%%%%%%%%%%%%%%%%%%%%%%%%%%%%%%%%
By combining the results of Proposition 1 and Proposition 2, and by employing the matrix inversion lemma \cite{horn2012matrix}, we obtain the following theorem about the proposed objective function $u(\S)$.
\begin{theorem}
The utility set function $u(\S)$ is a monotone, weak submodular function, $u(\emptyset) = 0$, and
\begin{align}
u_{(i,j,k)}(\S) = f_{(i,j,k)}(\S) + \gamma g_{(i,j,k)}(\S).
\label{eq:u_marginal}
\end{align}
\end{theorem}
\begin{proof}
First note that it clearly holds that $u(\emptyset) = f(\emptyset) +\gamma g(\emptyset) = 0$. Furthermore, since $u(\S) = f(\S)+ \gamma g(\S)$ is the sum of a monotone weak submodular and a submodular function, it is also monotone weak submodular and $\mathcal{C}_u \leq (2\frac{\lambda_M }{\lambda_m })^3$. 
Finally, we introduce $g_{(i,j,k)}(\S)$ in \ref{eq:g_marginal} and recursively find $f_{(i,j,k)}$ as
\begin{align}
f_{(i,j,k)}(\S) &= \frac{\vh_{j_k}^{\top}\F_{i,\S}^{-2}\vh_{j_k}}{\sigma_j^2 + \vh_{j_k}^{\top} \F_{i,\S}^{-1} \vh_{j_k}}, \\
\F_{i,\S \cup (i,j,k)}^{-1} &= \F_{i,\S}^{-1} - \frac{\F_{i,\S}^{-1}\vh_{j_k}\vh_{j_k}^{\top}\F_{i,s}^{-1}}{\sigma_j^2 + \vh_{j_k}^{\top}\F_{i,\S}^{-1}\vh_{j_k}},
\label{eq:f_marginal}
\end{align}
where %$\F_{i,\S}$ is
\begin{align*}
\F_{i,\S} = \F_{i,t} + \sum_{(i,j,k)\in \S} \frac{1}{\sigma_j^2} \vh_{j_k}\vh_{j_k}^{\top}.
\end{align*}
\end{proof}
%%%%%%%%%%%%%%%%%%%%%%%%%%%%%%%%%%%%%%%%%%%%%%%%%%%%%%%%
%%%%%%%%%%%%%%%%%%%%%%%%%%%Algorithm%%%%%%%%%%%%%%%%%%%%
%%%%%%%%%%%%%%%%%%%%%%%%%%%%%%%%%%%%%%%%%%%%%%%%%%%%%%%%
\section{Greedy Exchange of Observations}\label{sec:alg} 
The analysis of combinatorial characteristics of the proposed utility set function reveals that the optimization problem in \ref{eq:balance} is that of maximizing a monotone weak submodular set function subject to cardinality constraint. Therefore, in order to find a near-optimal scheduling of the observations exchange, we resort to their greedy selection. More specifically, at each time step $t$, the relay node observes the performances of the local nodes and calculates the marginal gain of the possible distribution patterns $(i,j,k)$ using \ref{eq:u_marginal}. Then it adds the pattern yielding the highest marginal gain to the scheduling set $\S_t$ and updates the performance records $\F_{i,\S}$ for each node using \ref{eq:f_marginal}. After repeating this procedure $K$ times, the relay node sends the instructions for the exchange of observations to the individual nodes.

%=================================== ALGORITHM 1
\renewcommand\algorithmicdo{}	% removes "DO" from for loops
\begin{algorithm}[t]
\caption{Greedy Observation Distribution Scheduling}
\label{alg:greedy}
\begin{algorithmic}[1]
    \STATE \textbf{Input:}  $\P_{i,t-1}$, $\vH_{i,t}$, $K$, for $i=1,\dots,m$.
    \STATE \textbf{Output:} Subset $\S_t\subseteq \X$ with $|S_t|=K$.
    \STATE Initialize $\S_t =  \emptyset$, $\F_{i,\S}=\F_{i,t}$ for $i=1,\dots,m$.
	\FOR{$k = 1,\dots, K$}
            \STATE $(i,j,k) = \argmax_{(i',j',k')\notin \S_t} u_{(i',j',k')}(\S_t)$.\vspace{0.2cm}
            \STATE Update $\S_t \leftarrow \S_t \cup \{(i,j,k)\}$.\vspace{0.2cm}
            \STATE Update $\F_{i,\S}^{-1} \leftarrow \F_{i,\S}^{-1}-\frac{\F_{i,\S}^{-1}\vh_{j_k}\vh_{j_k}(t)^\top\F_{i,\S}^{-1}}{\sigma_j^2+\vh_{j_k}^\top\F_{i,\S}^{-1}\vh_{j_k}}$. \vspace{0.2cm}
	\ENDFOR
	\RETURN $S_t$.
\end{algorithmic}
\end{algorithm}
%===================================

The proposed method is formalized as Algorithm \oldref{alg:greedy}. Performance and complexity of the greedy algorithm are characterized by the following theoretical results.
\begin{theorem}
Let $\mathcal{C}_u$ be the maximum element-wise curvature of $u(\S)$, i.e., the objective function of the balanced performance promoting scheduling problem in \ref{eq:balance}. Let $\{\sO_i\}_{(i,-,-) \in \S}$ denote the set $\S$ of the observations selected to be communicated through the network by Algorithm 1 at time $t$, and let $\{\sO_i^{\ast}\}_{(i,-,-) \in \S^{\ast}}$ be the optimal schedule
of \ref{eq:balance} such that $\sum_{(i,-,-) \in \S} |\sO_i| \leq K$, and $\sum_{(i,-,-) \in \S^{\ast}}|\sO_i^{\ast}| \leq K$. Then, it holds that
\begin{equation}
u(\S) \geq (1-e^{-\frac{1}{c}})u(\S^\ast),
\end{equation}
where $c = \max\{1, \mathcal{C}_u\}$. Furthermore, the computational complexity of Algorithm \oldref{alg:greedy} is $\mathcal{O}(Kmn^2\sum_{i=1}^m{|\L_{i,t}|})$.
\end{theorem}
%%%%%%%%%%%%%%%%%%%%%%%%%%%%%%%%%%%%%%%%%%%%%%%%%%%%%%%%
%%%%%%%%%%%%%%%%%%%%%%%Simulations%%%%%%%%%%%%%%%%%%%%%%
%%%%%%%%%%%%%%%%%%%%%%%%%%%%%%%%%%%%%%%%%%%%%%%%%%%%%%%%
\section{Simulation Results}\label{sec:sim}
% \begin{figure}[!h]
% \centering
% \includegraphics[width=3.5in]{Figures/rateVSglobal.eps}
% \label{sim:rateVSglobal}
% \end{figure}

% \begin{figure}[!h]
% \centering
% \includegraphics[width=3.5in]{Figures/nodeVise_.eps}
% \label{sim:nodeVise}
% \end{figure}
In this section, we study the performance of the proposed algorithm in different scenarios. In particular, we simulate a fully connected network having 3 nodes, set the dimension of the state vector to $n=50$,
and assume that a relay node is given information about the observation matrices of the individual nodes. For the state-transition matrix of the linear dynamical system, we set $\vA_t = 0.8\vI_n$ and randomly generate partial observation matrices $\vH_{i,t}$. The observation patterns of the nodes vary with different runs; however, we preserve the rank of the matrices -- in particular, $\text{rank}(\vH_{1,t}) = 21$, $\text{rank}(\vH_{2,t})=37$ and $\text{rank}(\vH_{3,t})=5$. We assume a zero-mean Gaussian process noise and a zero-mean Gaussian observation noise at individual nodes with covariance matrices $\vQ = 0.2\vI_n$ and $\vR_{i,t}=0.05\vI_n$, respectively.
We run 10 Monte-Carlo simulations and select time horizons for each run as $T=20$.

We first consider the MSE performances of the individual nodes in the network under regularized ($\gamma > 0$) and non-regularized ($\gamma = 0$) settings. The total number of measurements that can be exchanged among the units is set to $K=40$. The regularization coefficients are set to $\gamma=200$ and $\gamma=0$; the large difference between the regularization coeffcients will emphasize the effect of the balancing term on the individual node performance. In Fig. \oldref{fig:nodeVise}, we observe that the regularization term balances the individual node performances. We also observe that in the absence of regularization the nodes exhibit temporally rapidly varying MSE performance. This is primarily due to a deterministic nature of the greedy selection of the set of observations shared among the agents. In particular, when the regularization term is set to zero, at each time step the algorithm greedily schedules most of the observations to the node with the highest MSE. On the other hand, the non-zero regularization term ensures a temporally smoother and balanced MSE performances of the individual units.

To study the effect of the number of shared observations, we vary $K$ from $20$ to $100$. We compare the total network MSE (the sum of individual MSEs) at the last time step of the regularized and non-regularized schemes in Fig. \oldref{fig:rateVSglobal}. We observe that the non-regularized scheme always yields a lower total MSE as compared to the regularized one; this is expected since it completely focused on minimization of the total MSE and ignores balance of the individual units performances. For $K=100$, we observe that both networks essentially perform the same, which is expected due to sharing essentially all the observations in the network.

Finally, we investigate the effect of the regularization parameter $\gamma$ on balancing the individual performances of the nodes in the network. We set $K=40$ and vary $\gamma$ for the regularized network from $0.1$ to $100$ with log-scale increments. We compare the sum of pairwise MSE distances of the nodes in the regularized and non-regularized networks in Fig \oldref{fig:gammaVSglobal}. We observe that the use of higher regularization coefficients results in a more balanced performances between individual nodes.
%===================================
\begin{figure}[t]
\centering
    \includegraphics[width=0.49\textwidth]{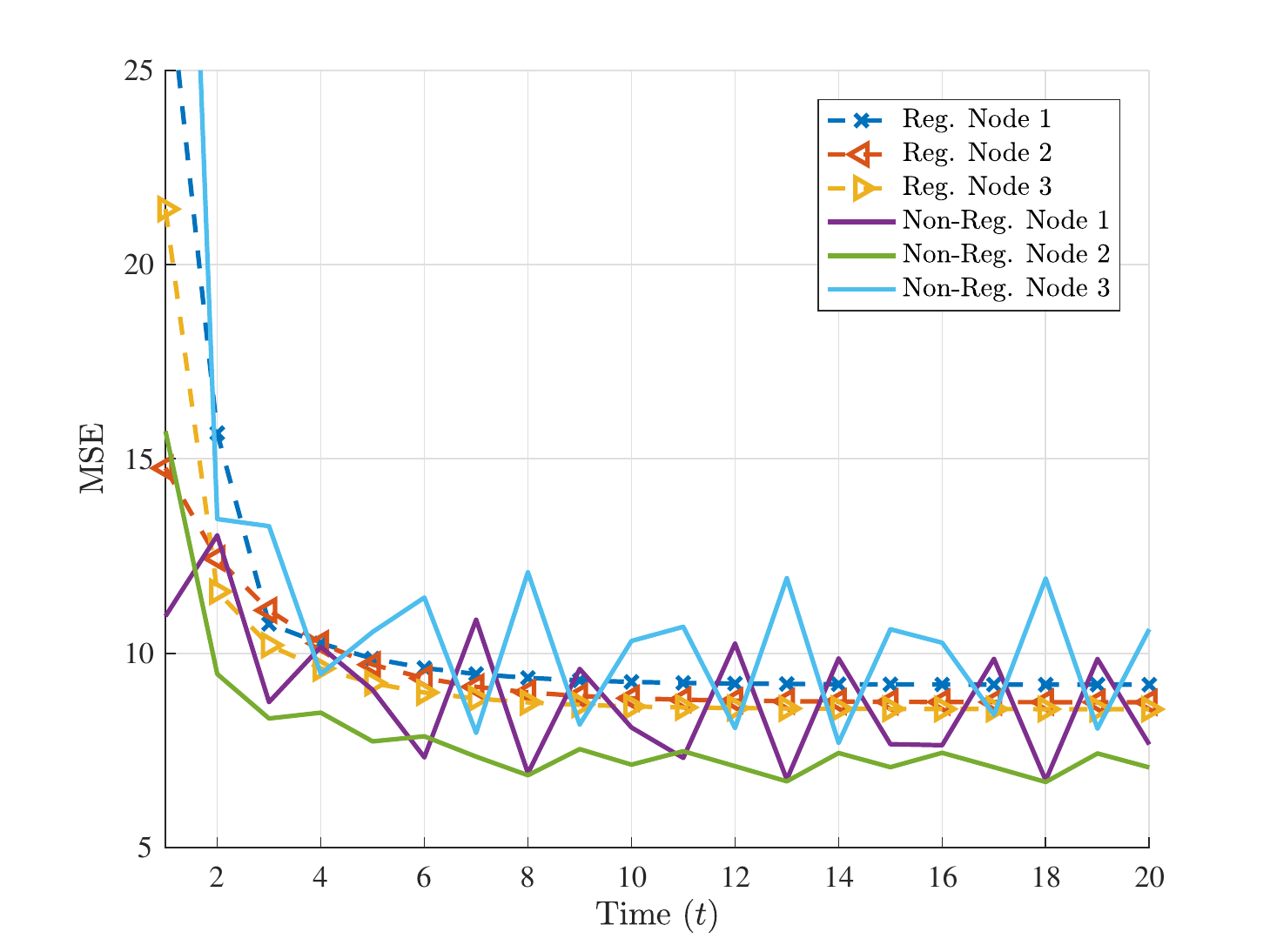}
    \caption{The MSEs of the individual units for the balanced
    performance promoting (labeled as "Reg.") and unbalanced
    (labeled as "Non-Reg.") measurement exchange schemes. The
    network contains three sensing units performing distributed
    state estimation by means of Kalman filtering. }
\label{fig:nodeVise}
\vspace{-0.6cm}
\end{figure}
%===================================
\begin{figure}[t]
\centering
    \includegraphics[width=0.49\textwidth]{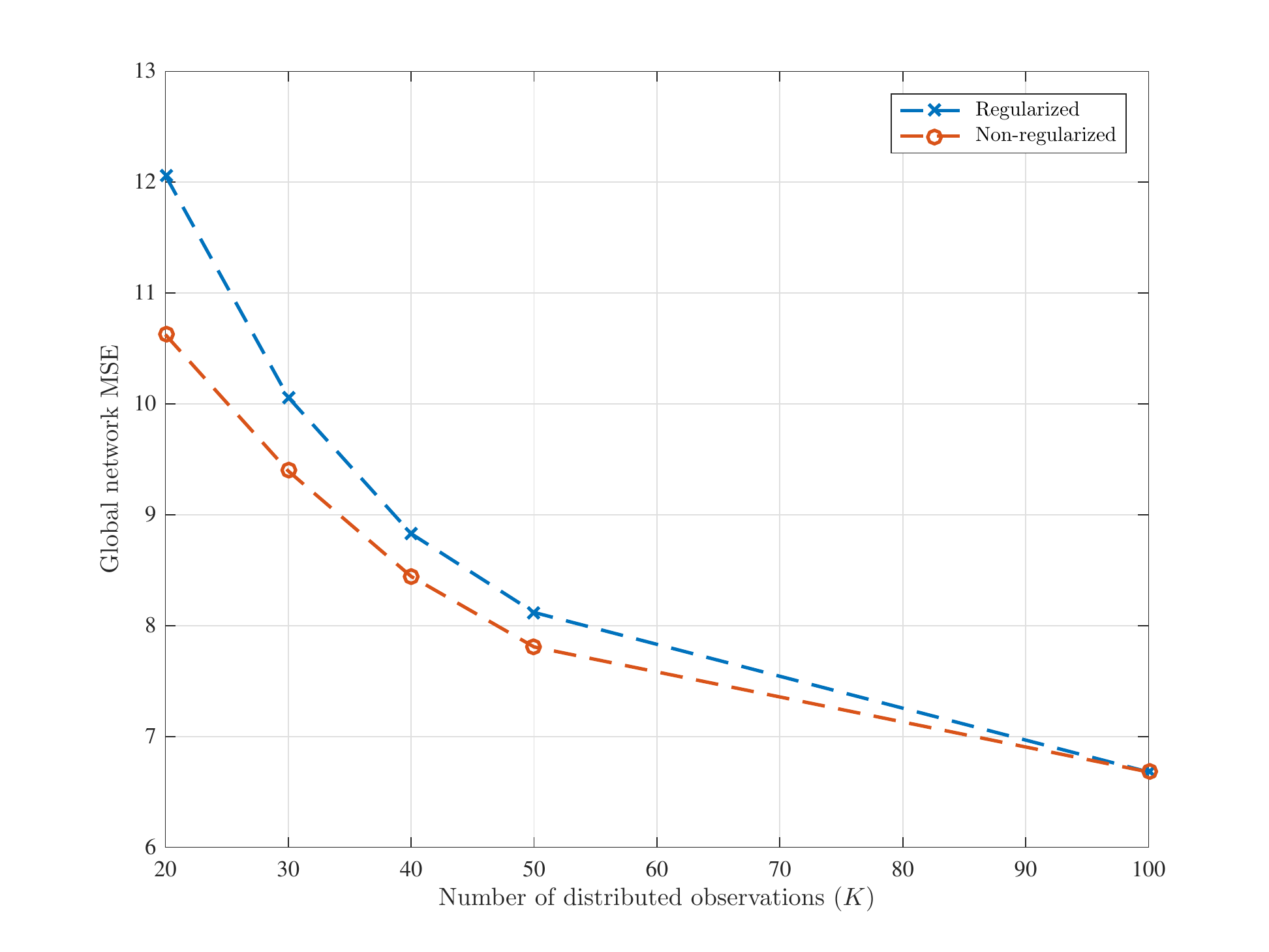}
    \caption{A comparison of the total MSE of the proposed balance-promoting and unbalanced measurement sharing schemes versus the number of shared observations $K$. As $K$ increases, the gap between the two schemes is reduced.}
\label{fig:rateVSglobal}
\vspace{-0.45cm}
\end{figure}
%===================================
\begin{figure}[t]
\centering
    \includegraphics[width=0.49\textwidth]{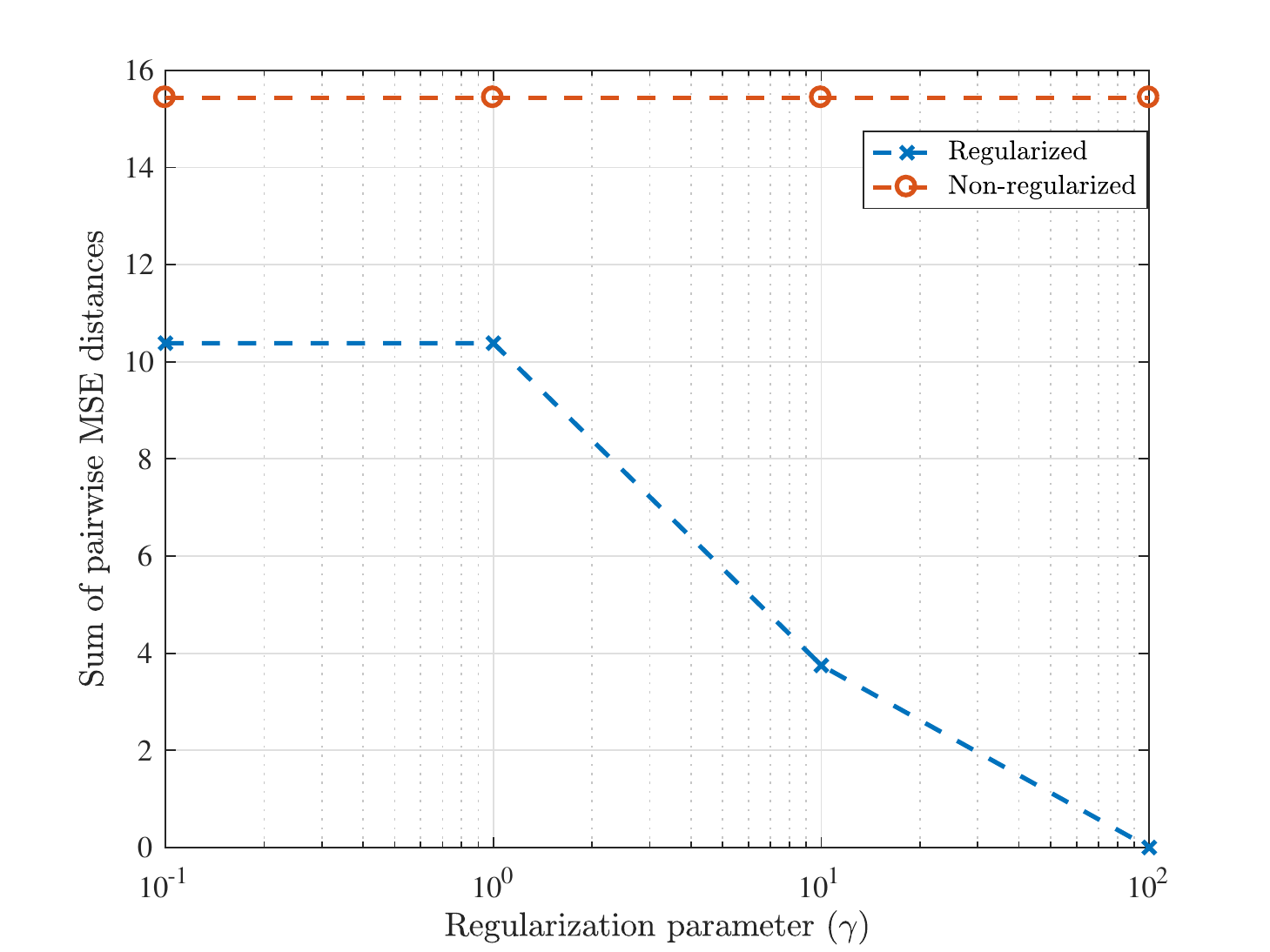}
    \caption{A comparison of sum of pairwise node-level MSE distances for varied regularization parameter $\gamma$. As $\gamma$ increases, the proposed balance-promoting framework attempts to decrease the MSE distances of the individual nodes across the network.}
\label{fig:gammaVSglobal}
\vspace{-0.5cm}
\end{figure}
%===================================
% %%%%%%%%%%%%%%%%%%%%%%%%%%%%%%%%%%%%%%%%%%%%%%%%%%%%%%%%
% %%%%%%%%%%%%%%%%%%%%%%%Conclusion%%%%%%%%%%%%%%%%%%%%%%%
% %%%%%%%%%%%%%%%%%%%%%%%%%%%%%%%%%%%%%%%%%%%%%%%%%%%%%%%%
\section{Conclusion} \label{sec:concl}

In this paper, we considered the task of distributed state estimation in a communication-constrained network of sensing units. The network consists of units with sensing and communication capabilities as well as a relay center that schedules the exchange of information in the network. In addition to minimizing the total mean-square error, a certain level of performance balancing is desired throughout the network. We formulated this task as that of maximizing a monotone objective function subject to cardinality constraint. The proposed objective function is the sum of two monotone set functions: the first function, that is weak submodular, is inversely related to the total MSE of the network while the second one is submodular and favors a schedule of observation exchange that promotes balanced performance of individual units. Since the proposed formulation is NP-hard, we developed a simple greedy algorithm and theoretically analyzed its performance by deriving a constant factor approximation on its achievable utility as compared to the utility attained by the optimal schedule. Through a series of simulations, we demonstrated that the proposed formulation minimizes the total MSE of the network while balancing individual units performance. As part of the future work, we will extend the proposed framework to a network of potentially nonlinear dynamical systems with multiple relay centers and analyze its performance.

%%%%%%%%%%%%%%%%%%%%%%%%%%%%%%%%%%%%%%%%%%%%%%%%%%%%%%%%
%%%%%%%%%%%%%%%%%%%%%%%References%%%%%%%%%%%%%%%%
%%%%%%%%%%%%%%%%%%%%%%%%%%%%%%%%%%%%%%%%%%%%%%%%%%%%%%%%
%\newpage
\bibliographystyle{ieeetr}\footnotesize
\bibliography{refs}

% %%%%%%%%%%%%%%%%%%%%%%%%%%%%%%%%%%%%%%%%%%%%%%%%%%%%%%%%
%\newpage
% \normalsize
% \begin{appendices}
% \input{Appendix}
% \end{appendices}
\end{document}